\documentclass[11pt]{article}
\usepackage{amsthm,amsfonts,amsmath}
\usepackage{graphicx,graphics}
 \newtheorem{theorem}{Theorem}[section]
 \newtheorem{lemma}[theorem]{Lemma}
 \newtheorem{proposition}[theorem]{Proposition}
 \newtheorem{corollary}[theorem]{Corollary}

\newcommand{\upchi}{\raise1pt\hbox{$\chi$}}
\newcommand{\R}{{\mathord{\mathbb R}}}

\newcommand{\N}{{\mathord{\mathbb N}}}

\renewcommand{\|}{{\Vert}}
\numberwithin{equation}{section}  

\begin{document}
\title{Optimal Concentration for $SU(1,1)$ Coherent State Transforms and An Analogue of the Lieb-Wehrl Conjecture for $SU(1,1)$}

 \author{
 Jogia Bandyopadhyay\footnote{E-mail: gtg110n@mail.gatech.edu}
 \\
 {\footnotesize Department of Physics, Georgia Institute of Technology, Atlanta, GA 30332, USA}
 }

\vspace{2pt}

\date{}

\maketitle
\begin{abstract}
We derive a lower bound for the Wehrl entropy in the setting of
$SU(1,1)$. For asymptotically high values of the quantum number $k$,
this bound coincides with the analogue of the Lieb-Wehrl conjecture
for $SU(1,1)$ coherent states. The bound on the entropy is proved
via a sharp norm bound. The norm bound is deduced by using an
interesting identity for Fisher information of $SU(1,1)$ coherent
state transforms on the hyperbolic plane $\mathbb{H}^{2}$ and a new
family of sharp Sobolev inequalities on $\mathbb{H}^{2}$. To prove
the sharpness of our Sobolev inequality, we need to first prove a
uniqueness theorem for solutions of a semi-linear Poisson equation
(which is actually the Euler-Lagrange equation for the variational
problem associated with our sharp Sobolev inequality) on
$\mathbb{H}^{2}$. Uniqueness theorems proved for similar semi-linear
equations in the past do not apply here and the new features of our
proof are of independent interest, as are some of the consequences
we derive from the new family of Sobolev inequalities.
\end{abstract}

\renewcommand{\thefootnote}{}
\footnote{Work partially supported by U.S. National Science
Foundation grant DMS 06-00037}

 \section{Introduction}

 Let $M$ be a  Riemannian manifold with volume element  $ {\rm d}\mathcal{M}$.
 For a probability density $\rho$ on $M$, that is, a non negative measurable
 function on $M$ with $\int_M\rho d\mathcal{M} =1$,
 its  entropy is defined as:
 \begin{eqnarray}\label{1.1}
 S(\rho)= -\int_M  \rho \ln\rho \  {\rm d}\mathcal{M}
  \end{eqnarray}

 Thus defined, the entropy of a density $\rho$ can be thought of as a measure of
 its ``concentration''.   If some part of the mass of $\rho$ is very nearly concentrated in
 a multiple of a Dirac mass, then $S(\rho)$ will be very negative.
 We shall be mainly interested in the case in which $M$ is the phase space of
 some classical system, so that, in particular, $M$ is a symplectic manifold.
 In that case,  we shall refer to $\rho$ as a {\it classical density}, and $S(\rho)$
 as its {\it classical entropy}.

 The uncertainty principle limits the extent of possible concentration in phase space:
 For instance it prevents both the momentum variables $p$ and the configuration
 variables $q$ from taking on well-defined values at the same time.
 Indeed, a quantum mechanical density $\rho_{Q}$ is a non negative operator on
the Hilbert space $ \mathcal{H}$, which is the state space of the
quantum system, having unit trace. Then the quantum entropy (or von
Neuman entropy) of  $\rho^{Q}$ is defined by
 \begin{eqnarray}\label{1.2}
 S^{Q}(\rho^{Q})=-\mbox{Tr }\rho^{Q}\ln\rho^{Q}\ .
 \end{eqnarray}
Since all of the eigenvalues of $\rho^Q$ lie in the interval
$[0,1]$, it is clear that
\begin{eqnarray}\label{qb}
S^Q(\rho^Q) \ge 0\ .
 \end{eqnarray}

As Wehrl emphasized [Weh], when one considers a quantum system and
its corresponding classical analogue, not all of the classical
probability densities on the phase space $M$ can correspond to
physical densities for the quantum system, and  one might expect a
lower bound on the classical entropy of those probability densities
that do correspond to actual quantum states.

There is a natural way to make the correspondence between quantum
states and classical probability densities on phase space, which
goes back to Schr\"odinger. It is based on the {\it coherent state
transform}, which is an isometry $\mathcal{L}$ from the quantum
state space $\mathcal{H}$ {\it into} $L^2(M)$, the Hilbert space of
square integrable functions on the classical phase space $M$.  Since
it is an isometry, if $\psi$ is any unit vector in $\mathcal{H}$,
$$\rho_\psi = |\mathcal{L}\psi|^2$$
is a probability density on $M$. Wehrl [Weh] proposed defining the
classical entropy of a quantum state $\phi$ in this way (note the
the corresponding density matrix is rank one, and hence the von
Neuman entropy would be zero). The Wehrl entropy is defined in terms
of the coherent states
 for the quantum system and is bounded below by the quantum entropy.
 It has several physically desirable features such as monotonicity,
 strong subadditivity, and of course, positivity.

Wehrl identified the class of probability densities arising through
the coherent state transform as the class of {\it quantum
mechanically significant} probability densities on $M$, and
conjectured that corresponding to (\ref{qb}), there should be a
lower bound on $S( |\mathcal{L}\psi|^2)$ as $\psi$ ranges over the
unit sphere in $\mathcal{H}$.

Specifically, when $\mathcal{H}$ is $L^2(\R,{\rm d}x)$, so that the
classical phase space is $\R^2$ with its usual symplectic and
Riemannian structure, Wehrl conjectured that the lower bound on  $S(
|\mathcal{L}\psi|^2)$ is attained when $\psi$ is a minimal
uncertainty state $\psi_{\rm min}$, also known as a {\it Glauber
coherent state}. That is:
\begin{eqnarray}\label{cslb}
\inf_{\|\psi\|_\mathcal{H} =1}S( |\mathcal{L}\psi|^2) = S(
|\mathcal{L}\psi|^2_0)\ .
\end{eqnarray}
This was proved by Lieb [Lie] . There is a natural analogue of the
Wehrl cojecture for other state spaces, and other coherent state
transforms. Lieb generalized the Wehrl conjecture
 to the $SU(2)$ coherent
 states, for which the corresponding  classical phase space is  $S^2$,
 the two-dimensional sphere, with its usual Riemannian and symplectic structure.
 The analogues of the Glauber coherent states in this case are the
 {\it Bloch coherent states} generated by least weight vectors in the various
 unitary representations of $SU(2)$, and Lieb conjectured the analogue of  (\ref{cslb})
 for the $SU(2)$ coherent state transform.

 Although  Lieb's conjecture for $SU(2)$  is still open, it has attracted the attention of a number of researchers,
 and much progress has been made.  The various unitary representations of $SU(2)$ are indexed by
 a half integer $j$, which is the {\it quantum number} in this context; for each such $j$ there is a coherent state
 transform, and hence a conjectured lower bound of the Wehrl entropy. The bound is
 trivial for $j=1/2$, in which case every state is a Bloch coherent state, but is is already
 non trivial for $j=1$. Schupp  [Sch] proved the conjecture for
 $j=1$ and $j=3/2$.
 Later Bodmann [Bod] proved a result which may be
 seen as complementary to Schupp's result; he deduced a lower bound for the Wehrl
 entropy of $SU(2)$ coherent states, for which the high spin
 asymptotics coincided with Lieb's conjecture up to, but not
 including, terms of first and higher orders in the inverse of spin
 quantum number $j$.

 Bodmann did this by proving a sharp $L^p$ bound on the range of the coherent state transform.
 This led to a proof of an analogue of Lieb's conjecture for certain {\it Renyi entropies}:
 For any $p>1$ and any classical density $\rho$,  define
 \begin{eqnarray}\label{renp}
S_p(\rho) = \frac{1}{p-1}\ln\left(\|\rho\|_p\right)\ .
\end{eqnarray}
where $\|\rho\|_p$ is the $L^p$ norm of $\rho$.  Then it is easy to
see that
$$
\lim_{p\to 1}S_p(\rho) = S(\rho)\ .$$ Bodmann derived his bound on
Renyi entropies from a Sobolev type inequality and a Fisher
information identity, which is another type of concentration bound
on the range of the coherent state transform. The Fisher information
$I(\rho)$ of a probabilty density $\rho$ on $M$ is defined by
$$I(\rho) = \int_M |\nabla \ln\rho |^2\rho\  {\rm d}\mathcal{M} = 4  \int_M |\nabla \sqrt{\rho} |^2\ {\rm d}\mathcal{M}\ .$$
For the Glauber coherent state transform, Carlen [Car]had proved
that all classical densities on $\R^2$ arising through the coherent
state transform had the {\it same} finite value of the Fisher
information. He then used that together with the logarithmic Sobolev
inequality (cf. [Gro]) to give a new proof of Wehrl's conjecture,
and to show that the lower bound in (\ref{cslb}) is attained only
for Glauber coherent states. Bodmann proved an analogue of Carlen's
result for Fisher information, and used this, together with a sharp
Sobolev inequality, instead of the sharp logarithmic Sobolev
inequality, to obtain his Renyi information bounds.

In this paper, we investigate the analogue of the Lieb-Wehrl
conjecture for $SU(1,1)$. The representations of $SU(1,1)$ belonging
to a discrete series, are labeled by a half-integer $k$, the
relevant {\it quantum number} in this context. While the classical
phase space for SU(2) is the sphere $S^2$, for $SU(1,1)$ the
classical phase space is $H^2$, the hyperbolic plane. It is natural
to conjecture that, here too, the coherent states generated by the
least-weight vector of the representation provide a lower bound on
the entropy, as in Lieb's conjecture for $SU(2)$. We prove that this
is indeed asymptotically true, in the semi-classical limit.   We
also prove that this is exactly true if one replaces the entropy by
an appropriate Renyi entropy. To obtain these results, we prove a
number of theorems concerning  analysis in $H^2$ that are of
independent interest. Specifically, we prove a new sharp Sobolev
inequality, and a sharpened energy--entropy inequality in $H^2$. The
Sobolev inequality is
 \begin{eqnarray*}
\| f\|_{q}^{q} +
\frac{4}{kq(kq-2)}\int|\nabla|f|^{q/2}|^{2}\geq\left(\frac{2k-1}{kq-1}\right)\left(\frac{kp-1}{2k-1}\right)^{q/p}{\left(\frac{kq-1}{kq-2}\right)\|f\|_{p}^{q}}
\end{eqnarray*}
 where $p=q+1/k$, $q\geq 2$, $kq> 2$,
 and we determine all of the cases of equality.

To prove the sharpness of our Sobolev inequality we need to prove
and use a uniqueness result for radial solutions of a semi-linear
Poisson equation on the hyperbolic plane.  The nature of this
equation on $H^2$ is substantially different from that of similar
equations which have been investigated in the past. The methods
developed here may well be useful for other uniqueness problems.

We then prove the following Fisher information identity:
\begin{eqnarray*}
 \int|\nabla|\mathcal{L}\psi|^{q/2}|^{2}=\frac{1}{4}kq\int|\mathcal{L}\psi|^{q}
 \end{eqnarray*}
 where  $q$ is a positive number such that $kq>2$. As mentioned above,
 an identity like
 this was first proved by Carlen [Car] for coherent state transforms associated with the Glauber coherent
 states.

The sharp Sobolev inequality and the Fisher information identity
allow us to prove an $L^p$ norm estimate {\it a la} Bodmann. This
norm estimate is used to deduce a lower bound for the Wehrl entropy
of coherent state transforms via a convexity argument, and the
result is:

\begin{eqnarray*}
S(|\mathcal{L}\psi(\zeta)|^{2})\geq
2k\ln\left(1+\frac{1}{2k-1}\right)
\end{eqnarray*}
It is seen that for high values (this gives us the semi-classical
limit) of the quantum number $k$, this lower bound coincides with
the analogue of the Lieb-Wehrl conjecture, up to but not including
terms of first and higher order in $\displaystyle k^{-1}$.

The methods used to bound the entropy also serve to produce a new,
sharpened energy--entropy inequality for functions on $H^2$. An
energy--entropy inequality is an inequality of the form
\begin{eqnarray}\label{enen}
-S(\rho) \le \Phi_{M}(I(\rho))\
\end{eqnarray}
for some function $\Phi$. Since the Fisher information, can be
expressed in terms of an energy integral as shown above, the
entropy-energy terminology is natural. For a given Riemannian
manifold $M$, the entropy--energy problem is to determine the least
function $\Phi: \R_+ \to \R$ for which (\ref{enen}) is true.

For example, in the case $M = \R^2$, the optimal $\Phi$ is known:
\begin{eqnarray*}
-S(\rho) \leq \ln\left(\frac{4}{\pi e}I(\rho)\right)
\end{eqnarray*}
Equality is achieved when $\rho$ is an isotropic Gaussian function.
For an appropriate choice of the variance of the Gaussian, $I(\rho)$
can take on any value, and this inequality is sharp for all values
of $I(\rho)$.  That is,
$$\Phi_{\R^2}(t) = \ln\left(\frac{4}{\pi e}t\right)\ .$$

There has been much investigation of entropy-energy inequalities for
various Riemannian manifolds (see [Bec2], [Heb], [Rot] for example).
Though there has been significant progress, many questions are still
open.

 In the case of $\mathbb{H}^{2}$ , Beckner
proved [Bec2] that the entropy--energy inequality for $H^2$ holds
with the same $\Phi$  as in $\R^2$. That is,
$$\Phi_{H^2}(t) \le \Phi_{\R^2}(t)$$
for all $t$.

This result is asymptotically sharp in the sense that
$$\lim_{t\to 0} \frac{\Phi_{H^2}(t)}{ \Phi_{\R^2}(t)} =1\ ,$$
however,  $\Phi_{H^2}(t) < \Phi_{\R^2}(t)$. We shall give sharpened
estimates on $\Phi_{H^2}(t)$.

The paper is organized as follows: in Section 2 we give a
description of a discrete representation of $SU(1,1)$. We then
define the associated coherent states and coherent state transform.
Given any quantum state $\psi$, we denote its coherent state
transform by $\mathcal{L}\psi(\zeta)$, where the complex number
$\zeta$ is used to label the coherent states. We show that these
coherent state transforms are actually probability densities on the
hyperbolic plane. We also state the analogue of the Lieb-Wehrl
conjecture in this setting.

Section 3 contains the proof of the lower bound for the Wehrl
entropy for SU(1,1), and the results leading up to it.
 Here we prove
Fisher information identity for the coherent state transforms, and
the sharp Sobolev inequality. The proof of the latter result uses
the uniqueness result that is postponed to the final section.

Section 4 contains the sharpened entropy--energy inequality for
$H^2$, and finally Section 5,  the longest one, contains our
uniqueness proof.

The problem of proving an analogue of the Lieb-Wehrl conjecture in
the $SU(1,1)$ setting  was suggested to me by my advisor, Prof. Eric
Carlen. I am greatly indebted to him for introducing me to this
beautiful problem and helping me with many valuable suggestions and
discussions without which this work would not have been possible.

\setcounter{equation}{0}

\section{Representation of the group SU(1,1)
and the Construction of Coherent States} \vspace{.2in}

The group $SU(1,1)$ consists of  unimodular $2\times 2$ matrices
which leave the Hermitian form $|z_{1}|^{2}-|z_{2}|^{2}$ invariant.
These matrices can be parametrized by a pair of complex numbers,
$\alpha,\beta$ as follows:
\begin{eqnarray*}
g= \left( \begin{array}{ccc} \alpha &
\beta\\
\bar{\beta} & \bar{\alpha}
\end{array} \right), \qquad |\alpha|^{2}-|\beta|^{2}=1
\end{eqnarray*}
One can define a new variable $\displaystyle z=\frac{z_{2}}{z_{1}}$
and describe the action of the element $g\in SU(1,1)$ on
$\mathbb{C}^{1}$ as:
\begin{eqnarray*}
z\longrightarrow z_{g}=\frac{\alpha z+\bar{\beta}}{\beta
z+\bar{\alpha}}
\end{eqnarray*}

However, the group action on $\mathbb{C}^{1}$ is not transitive; in
fact the complex plane is foliated into three orbits, namely, i) the
interior of the unit disk, ii) the boundary of the unit disk, and,
iii) the complement of the closed unit disk in the complex plane.

It is easy to see [Per] that the set of elements of $SU(1,1)$ having
real, positive diagonal entries can be identified with the interior
of the unit disc, $\{\zeta:|\zeta|<1\}$. We would work with one of
the two discrete series of representations of $SU(1,1)$, in the
space of functions that are defined and analytical in the unit disc.

The Lie algebra for $SU(1,1)$ has three generators as its basis
elements, which we call $K_{0}, K_{1}$ and $K_{2}$ following
Perelomov. The commutation relations satisfied are:
\begin{eqnarray*}
[K_{1},K_{2}]=-iK_{0},\qquad [K_{2},K_{0}]=iK_{1},\qquad
[K_{0},K_{1}]=iK_{2}
\end{eqnarray*}
There is one Casimir operator given by:
$\hat{C}=K_{0}^{2}-K_{1}^{2}-K_{2}^{2}$. So for any irreducible
representation the operator is a multiple of the identity and we
write:
\begin{eqnarray*}
\hat{C}= k(k-1)\hat{I}
\end{eqnarray*}
 Thus a particular representation of $SU(1,1)$ is labeled by a
 single number $k$. For the discrete series this number takes on
 discrete half-integral values, $k=1/2,1,3/2,...$  [Bar].
 Let us call a particular representation $T^{k}(g)$. We choose the
 simultaneous eigenvectors of the Casimir operator $\hat{C}$ and $K_{0}$
 to be the basis vectors. We use Dirac's bra-ket notation and denote these vectors by $|k,\mu\rangle$
 where:
 \begin{eqnarray*}
 K_{0}|k,\mu\rangle=\mu|k,\mu\rangle
 \end{eqnarray*}
 Here $\mu=k+m$ and $m$ is either zero or any positive integer
 [Per] (the representations are infinite-dimensional).
 We now look at a realization of $T^{k}(g)$ in the space $\mathcal{G}_{k}$ of
 functions $f(z)$ which are analytic inside the unit circle and which satisfy the condition:
 \begin{eqnarray*}
\frac{2k-1}{\pi}\int_{D}|f(z)|^{2}(1-|z|^{2})^{2k-2}d^{2}z<\infty,\qquad
 D=\{z:|z|<1\}
 \end{eqnarray*}

 The invariant density for this realization of $T^{k}(g)$ is [Bar]:
 \begin{eqnarray*}
 d\varpi_{k}(z)=\frac{2k-1}{\pi}(1-|z|^{2})^{2k-2}d^{2}z
 \end{eqnarray*}
 The pre-factor $\displaystyle\frac{2k-1}{\pi}$ is chosen so that we have
 $\displaystyle(f,g)_{k}=\int_{D}\overline{f(z)}g(z)d\varpi_{k}(z)\equiv1$ when $f\equiv 1$ and $g\equiv 1$, where
 $(f,g)_{k}$ denotes the inner product of $f$ and $g$ in this
 representation. The group action on $\mathcal{G}_{k}$ in
 the multiplier representation $T^{k}(g)$ is given by [Bar]:
 \begin{eqnarray*}
 T^{k}(g)f(z)= (\beta z+\bar{\alpha})^{-2k}f(z_{g}), \qquad z_{g}=\frac{\alpha z+\bar{\beta}}{\beta
z+\bar{\alpha}}
\end{eqnarray*}

  The operators $T^{k}(g)$ with the group action
 defined as above furnish a unitary and irreducible representation of $SU(1,1)$ [Bar].
 Now, in $\mathcal{G}_{k}$ the generators act as first order
 differential operators. If, in stead of  the standard basis $K_{0}, K_{1}$ and $
 K_{2}$
 we switch to the ladder operators $K_{\pm}=\pm i(K_{1}\pm iK_{2})$
 and $K_{0}$, then we have [Per]:
 \begin{eqnarray*}
 K_{+}= z^{2}\frac{d}{dz}+2kz\quad,\qquad K_{-}=\frac{d}{dz}\quad,\qquad K_{0}=z\frac{d}{dz}+k
 \end{eqnarray*}
 From the form of $K_{0}$ it is clear that its eigenfunctions in
 this
 representation are monomials in $z$. Normalized with respect to
 the measure $d\varpi_{k}(z)$ these eigenvectors are written:
 \begin{eqnarray}\label{2.1}
 |k,k+m\rangle =
 \left(\frac{\Gamma(m+2k)}{m!\Gamma(2k)}\right)^{\frac{1}{2}}z^{m}
 \end{eqnarray}
To construct the coherent states let us choose the least-weight
vector $|k,k\rangle$ in $\mathcal{F}_{k}$. The stationary subgroup
for this state is the subgroup $H$ of diagonal matrices of the form

$\displaystyle h= \left( \begin{array}{ccc} e^{i\varphi/2} &
0\\
0 & e^{-i\varphi/2}
\end{array} \right)$. The factor space $G/H$ is realized as the unit
disk $\{\zeta: |\zeta|<1\}$, or equivalently, as the hyperbolic
plane $\mathbb{H}^{2}=\{\mathbf{n}:
|n|^{2}=n_{0}^{2}-n_{1}^{2}-n_{2}^{2}=1, n_{0}>0\}$ via the
following correspondence:
\begin{eqnarray*}
n_{0}=\cosh\frac{\tau}{2}\quad,\qquad
n_{1}=\sinh\frac{\tau}{2}\cos\phi\quad,\qquad
n_{2}=\sinh\frac{\tau}{2}\sin\phi\qquad\mbox{and   }\quad
\zeta=\tanh\frac{\tau}{2}e^{i\phi}
\end{eqnarray*}
An element of $G/H$  determines a hyperbolic rotation and we can
decompose the corresponding operator $T^{k}(g_{\mathbf{n}})$ as
follows:
\begin{eqnarray*}
T^{k}(g_{\mathbf{n}})=\exp
\left(\tanh\frac{\tau}{2}\exp(i\phi)K_{+}\right)\exp\left(-2\ln(\cosh\frac{\tau}{2})K_{0}\right)\exp\left(-\tanh\frac{\tau}{2}\exp(-i\phi)K_{-}\right)
\end{eqnarray*}
We let these operators act on the chosen least-weight vector
$|k,k\rangle$ to obtain an expression for the coherent states in
terms of the standard orthonormal basis vectors:
\begin{eqnarray*}
T^{k}(g_{\mathbf{n}})|k,k\rangle &=& \exp
\left(\tanh\frac{\tau}{2}\exp(i\phi)K_{+}\right)\exp\left(-2\ln(\cosh\frac{\tau}{2})K_{0}\right)\exp\left(-\tanh\frac{\tau}{2}\exp(-i\phi)K_{-}\right)|k,k\rangle\\
&=&(1-|\zeta|^{2})^{k}\sum_{m=0}^{\infty}\left(\frac{\Gamma(m+2k)}{m!\Gamma(2k)}\right)^{\frac{1}{2}}\zeta^{m}|k,m\rangle
\end{eqnarray*}
Represented as above, the coherent states are parametrized by a
complex number $\zeta$ on the unit disk or equivalently, by two real
parameters $\tau$ and $\phi$ on the hyperbolic $\mathbb{H}^{2}$. In
what follows, we will denote the coherent state corresponding to a
particular $\zeta$ by $|\zeta\rangle$. If we now choose any
arbitrary normalized vector
$|\psi\rangle=\sum_{m=0}^{\infty}a_{m}|k,m\rangle$, then we can
define its coherent state transform $\mathcal{L}\psi(\zeta)$ via the
following inner product:
\begin{eqnarray}\label{2.2}
\mathcal{L}\psi(\zeta)=\langle\psi|\zeta\rangle=(1-|\zeta|^{2})^{k}\sum_{m=0}^{\infty}\left(\frac{\Gamma(m+2k)}{m!\Gamma(2k)}\right)^{\frac{1}{2}}\bar{a}_{m}\zeta^{m}
\end{eqnarray}
Evidently $\mathcal{L}\psi(\zeta)$ is a function on the unit disk
and so the coherent state transform maps unit vectors in our
representation space $\mathcal{G}_{k}$ into functions on the unit
disk, which vanish at the boundary of the disk. This mapping becomes
an isometry if we equip the unit disk with the $L^{2}$-metric
corresponding to the measure:
$d\nu(\zeta)=\displaystyle\left(\frac{2k-1}{\pi}\frac{1}{(1-|\zeta|^{2})^{2}}\right)d^{2}\zeta$.
Note that $d\nu(\zeta)$ is just
$\displaystyle\left(\frac{2k-1}{4\pi}\right)$ times the standard
measure on the unit disk, that is, the measure $\displaystyle
d\mu(\zeta)=\left(\frac{4}{(1-|\zeta|^{2})^{2}}\right)d^{2}\zeta$,
obtained from the  Poincare metric on the disk. With inner product
defined in the usual way with respect to the measure $d\nu(\zeta)$,
the space of the coherent state transforms described above is a
Hilbert space [Bar]. We call this space $\mathfrak{F}_{k}$. The
transform $\mathcal{L}$ is thus an analogue of the Bargmann-Segal
transform for the Glauber coherent states based on the Heisenberg
group. Since $|\psi\rangle$ is a unit vector in our representation
space $\mathcal{G}_{k}$, its coherent state transform $\displaystyle
|\mathcal{L}\psi(\zeta)|^{2}d\nu(\zeta)$ is a probability density on
the unit disk. Thus, $\mathfrak{F}_{k}$ is a space of probability
densities on the unit disk. We can calculate the Wehrl entropy
$S(|\mathcal{L}\psi(\zeta)|^{2})$ associated with the coherent state
transform $\mathcal{L}\psi(\zeta)$. If the unit vector
$|\psi\rangle$ happens to be a coherent state itself, we find that:
$\displaystyle S(|\mathcal{L}\psi(\zeta)|^{2})=\frac{2k}{2k-1}$. The
analogue of the Lieb-Wehrl conjecture for $SU(1,1)$ coherent states
would then be:
\begin{proposition}
For all $\mathcal{L}\psi(\zeta)\in\mathfrak{F}_{k}$, the Wehrl
entropy is bounded below by:
\begin{eqnarray}\label{2.3}
S(|\mathcal{L}\psi(\zeta)|^{2})\geq \frac{2k}{2k-1}
\end{eqnarray}
\end{proposition}

\setcounter{equation}{0}

 \section{The Entropy Bound and Related Results}
\vspace{.2in}

In this section we first present a useful Fisher information
identity for functions in $\mathfrak{F}_{k}$, that relates the
$q$-norm (for all positive $q$ such that $kq>2$) of a function to
the $L^{2}$-norm of the associated gradient. We then prove a sharp
Sobolev inequality for functions in a larger function space
$\mathfrak{H}$, defined to be the space of bounded non-constant
functions $f\in W^{1,2}(D)$ on the unit disk which vanish at the
boundary; the norms here are computed with respect to the measure
$d\nu(\zeta)$. Next, we prove a sharp norm estimate for functions in
$\mathfrak{F}_{k}$ (note that $\mathfrak{F}_{k}$ is a subspace of
$\mathfrak{H}$) by converting the gradient norm of $|f|^{q/2}$ that
appears in our sharp Sobolev inequality, into the $L^{q}$-norm of
the function $f$, via the Fisher information identity. This sharp
norm estimate is then used to derive a lower bound on the entropy of
functions in $\mathfrak{F}_{k}$.

 The variational
problem associated with our sharp Sobolev inequality in the function
space $\mathfrak{H}$, naturally leads us to an Euler-Lagrange
equation which is actually a semi-linear Poisson equation on the
unit disk. We reduce the Euler-Lagrange equation to an ordinary
differential equation by using radially symmetric decreasing
rearrangements of functions. To prove the sharpness of the Sobolev
inequality we need to prove that the ground state solution, that is
to say, the solution that decays to zero at the boundary of the
disk, is unique. Since the proof is somewhat involved, we present a
detailed analysis of the Euler-Lagrange equation and relevant
results in section 5.

\subsection{A Fisher Information Identity} The Fisher information of a probability density function is a measure of its
concentration. In this subsection we prove a Fisher information
identity for functions in $\mathfrak{F}_{k}$.
 \begin{theorem}
For $\mathcal{L}\psi(\zeta)$ in
 $\mathfrak{F}_{k}$ the following identity holds:
 \begin{eqnarray*}
 \int|\nabla|\mathcal{L}\psi(\zeta)|^{q/2}|^{2}d\nu(\zeta)=\frac{1}{4}kq\int|\mathcal{L}\psi(\zeta)|^{q}d\nu(\zeta)
 \end{eqnarray*}
 where  $q$ is a positive number such that $kq>2$.
 \end{theorem}

 \begin{proof}

 Using the expression (\ref{2.2}) for the coherent state transforms in
 $\mathfrak{F}_{k}$, we can write:
\begin{eqnarray*}
 |\mathcal{L}\psi(\zeta)|^{q/2}=(1-|\zeta|^{2})^{kq/2}\left|\sum_{m=0}^{\infty}\left(\frac{\Gamma(m+2k)}{m!\Gamma(2k)}\right)^{\frac{1}{2}}\bar{a}_{m}\zeta^{m}\right|^{q/2}=(1-|\zeta|^{2})^{kq/2}|\Phi(\zeta)|
\end{eqnarray*}
where $\Phi(\zeta)$ is holomorphic in $\zeta$. Thus $\Phi(\zeta)$
satisfies the Cauchy-Riemann equations on the unit disk/hyperbolic
plane. Let us do our computations in terms of the radial variable
$\tau$ and the angular variable $\phi$ on the two-dimensional
hyperbolic plane. The gradient is then given by: $\displaystyle
\nabla=\left(\frac{\partial}{\partial\tau},\frac{1}{\sinh\tau}\frac{\partial}{\partial\phi}\right)$

A brief computation yields the following Cauchy-Riemann equations
for an analytic function $\Phi=u+iv$ on the hyperbolic plane:

\begin{eqnarray*}
 \frac{\partial u}{\partial\tau}=\frac{1}{\sinh\tau}\frac{\partial
v}{\partial\phi}\quad,\quad \frac{\partial
u}{\partial\phi}=-\sinh\tau\frac{\partial v}{\partial\tau}
\end{eqnarray*}
Using these two equations we obtain the following:
\begin{displaymath} \nabla u\cdot\nabla v=\frac{\partial
u}{\partial\tau}\frac{\partial
v}{\partial\tau}+\frac{1}{\sinh^{2}\tau}\frac{\partial
u}{\partial\phi}\frac{\partial v}{\partial\phi}=0\end{displaymath}
\begin{displaymath}
|\nabla u|^{2}=|\nabla v|^{2}
\end{displaymath}

We now compute some results for the non-holomorphic pre-factor in
the expression for the coherent state transforms.
\begin{displaymath}
\nabla(1-|\zeta|^{2})^{kq/2}=\left(\frac{\partial}{\partial\tau},\frac{1}{\sinh\tau}\frac{\partial}{\partial\phi}\right)(1-\tanh^{2}\frac{\tau}{2})^{kq/2}
=\left(-\frac{kq}{2}\tanh\frac{\tau}{2}\mbox{ sech
}^{kq}\frac{\tau}{2},0\right)
\end{displaymath}
Also,
\begin{eqnarray*}
\triangle(1-|\zeta|^{2})^{kq/2}=\left(\frac{\partial^{2}}{\partial\tau^{2}}+\coth\tau\frac{\partial}{\partial\tau}\right)\left(1-\tanh^{2}\frac{\tau}{2}\right)^{kq/2}
=\left(\frac{kq}{2}\right)^{2}\tanh^{2}\frac{\tau}{2}\mbox{ sech
}^{kq}\frac{\tau}{2}- \frac{kq}{2} \mbox{ sech }^{kq}\frac{\tau}{2}
\end{eqnarray*}

As for the holomorphic part of the transform, the Cauchy-Riemann
equations guarantee that:
\begin{displaymath}
\triangle|\Phi|^{2}=4|\nabla|\Phi||^{2}
\end{displaymath}
Thus:
\begin{eqnarray*}
 &&|\nabla|\mathcal{L}\psi(\zeta)|^{q/2}|^{2}\\
 &=& (1-|\zeta|^{2})^{kq}|\nabla|\Phi||^{2}+|\nabla(1-|\zeta|^{2})^{kq/2}|^{2}|\Phi|^{2}+2(1-|\zeta|^{2})^{kq/2}\nabla(1-|\zeta|^{2})^{kq/2}\cdot|\Phi|\nabla|\Phi|\\
&=&(1-|\zeta|^{2})^{kq}|\nabla|\Phi||^{2}+\frac{1}{4}|\Phi|^{2}(1-|\zeta|^{2})^{-kq}|\nabla(1-|\zeta|^{2})^{kq}|^{2}
+\frac{1}{2}\nabla(1-|\zeta|^{2})^{kq}\cdot\nabla|\Phi|^{2}\\
 &=&(1-|\zeta|^{2})^{kq}|\nabla|\Phi||^{2}+\frac{1}{4}|\Phi|^{2}\left(\triangle(1-|\zeta|^{2})^{kq}+kq(1-|\zeta|^{2})^{kq}\right)\\
 &\quad&+\frac{1}{2}\left(\nabla\cdot((1-|\zeta|^{2})^{kq}\nabla|\Phi|^{2})-(1-|\zeta|^{2})^{kq}\triangle|\Phi|^{2}\right)
\end{eqnarray*}
We notice that the divergence term, when integrated with respect to
the invariant measure $d\nu(\zeta)$
yields a vanishing surface integral for $\displaystyle kq>2$. Also,
\begin{displaymath}
\frac{1}{4}|\Phi|^{2}\triangle(1-|\zeta|^{2})^{kq}=\frac{1}{4}\left(\nabla\cdot(|\Phi|^{2}\nabla(1-|\zeta|^{2})^{kq})-
\nabla|\Phi|^{2}\cdot\nabla(1-|\zeta|^{2})^{kq}\right)
\end{displaymath}
We can ignore the divergence terms coming from the expression above
again by the same logic as before and write
\begin{eqnarray*}
\frac{1}{4}\int|\Phi|^{2}\triangle(1-|\zeta|^{2})^{kq}d\nu(\zeta)
=\frac{1}{4}\int(1-|\zeta|^{2})^{kq}\triangle|\Phi|^{2}d\nu(\zeta)
\end{eqnarray*}
Putting these all together we finally arrive at:
\begin{eqnarray*}
\int|\nabla|\mathcal{L}\psi(\zeta)|^{q/2}|^{2}d\nu(\zeta)=\int(1-|\zeta|^{2})^{kq}\left(|\nabla|\Phi||^{2}-\frac{1}{4}\triangle|\Phi|^{2}\right)d\nu(\zeta)+\frac{1}{4}kq\int|\Phi|^{2}(1-|\zeta|^{2})^{2k}d\nu(\zeta)
\end{eqnarray*}
The first term on the right hand side in the equation above,
vanishes due to analyticity of $\Phi$ as we have already shown,
yielding the following identity:
\begin{eqnarray*}
\int|\nabla|\mathcal{L}\psi(\zeta)|^{q/2}|^{2}d\nu(\zeta)=\frac{1}{4}kq\int|\mathcal{L}\psi(\zeta)|^{q}d\nu(\zeta)
\end{eqnarray*}
\end{proof}
\subsection{A Sharp Sobolev inequality and a Norm Estimate} We now prove a sharp Sobolev inequality for functions in $\mathfrak{H}$.
\begin{theorem}
For all functions in $\mathfrak{H}$ the following inequality holds:
\begin{eqnarray}\label{3.1}
\| f\|_{q}^{q} +
\frac{4}{kq(kq-2)}\int|\nabla|f|^{q/2}|^{2}\geq\left(\frac{2k-1}{kq-1}\right)\left(\frac{kp-1}{2k-1}\right)^{q/p}{\left(\frac{kq-1}{kq-2}\right)\|f\|_{p}^{q}}
\end{eqnarray}
where $p=q+1/k$, $q\geq 2$, $kq> 2$ and the norms are computed with
respect to the measure $d\nu(\zeta)$; equality is obtained if and
only if the function $f$ is a coherent state.
\end{theorem}
\begin{proof}
Proving Theorem $3.2$ is equivalent to showing that the infimum of
the functional
\begin{displaymath}
I[f]=\frac{\| f\|_{q}^{q} +
\frac{4}{kq(kq-2)}\int|\nabla|f|^{q/2}|^{2}}{\left(\frac{kq-1}{kq-2}\right)\|f\|_{p}^{q}}
\end{displaymath}
is
$\displaystyle\left(\frac{2k-1}{kq-1}\right)\left(\frac{kp-1}{2k-1}\right)^{q/p}$.
Since we are in the function space $\mathfrak{H}$, the existence of
the minimum is obvious. Let us take a minimizing sequence
$\{f_{n}\}$. We can now perform a radially symmetric decreasing
rearrangement, since the gradient norm can only decrease under such
a rearrangement while the other norms in the functional stay
constant. So each function in the minimizing sequence is replaced by
its decreasing rearrangement. Functions in the new sequence
$\{f_{n}^{*}\}$ thus obtained also have bounded norms and gradient
norms. The sequence being monotone and bounded we can use Helly's
principle to obtain a convergent subsequence. Since the functions
are in $W^{1,2}$, the convergence is in the $s$-norm, for all finite
$s$, by Rellich-Kondrashov theorem. We thus need to show that in a
class of radially symmetric solutions the minimizer is unique. The
minimizer satisfies the following Euler-Lagrange equation for our
optimization problem:

\begin{eqnarray}\label{3.2}
\triangle u +kq(kq-2)[\gamma u^{1+\frac{2}{kq}}-u]=0
\end{eqnarray}
where $u=|f|^{q/2}$, $\triangle$ is the Laplacian on the hyperbolic
plane (or, equivalently, the unit disk), $\gamma>0$  is fixed by
choosing the $p$-norm of the function $f$. It is readily seen that
this Euler-Lagrange equation is solved by the coherent state:
$f=A(1-|\zeta|^{2})^{k}$ where $A$ is a constant determined by
fixing the $p$-norm. Since we are dealing with radial functions
only, (\ref{3.2}) is equivalent to an ordinary differential
equation. We now refer to section 5, where we prove in detail that
there is only one solution of this ODE, in the space of radially
symmetric functions on the unit disk, which decays to zero at the
boundary of the disk (or, equivalently, decays to zero as the radial
coordinate on the hyperbolic plane tends to infinity). On the basis
of this uniqueness result we can conclude that the coherent state
$f=A(1-|\zeta|^{2})^{k}$ is indeed the unique solution and hence
furnishes the minimum.
\end{proof}

This sharp Sobolev inequality, coupled with our Fisher information
identity, trivially yields the following corollary:
\begin{corollary}
For all functions in $\mathfrak{F}_{k}$ the following inequality
holds:
\begin{eqnarray}\label{3.3}
||f||_{q}^{q}\geq
\left(\frac{2k-1}{kq-1}\right)\left(\frac{kp-1}{2k-1}\right)^{q/p}||f||_{p}^{q}
\end{eqnarray}
where $q\geq 2$; equality is obtained if and only if the function
$f$ is a coherent state.
\end{corollary}
\begin{proof}
The Fisher information identity for functions in $\mathfrak{F}_{k}$
tells us:
\begin{eqnarray*}
\int|\nabla|f|^{q/2}|^{2}=\frac{1}{4}kq\int|f|^{q}
\end{eqnarray*}
We can thus re-write the left hand side of (\ref{3.1}) as:
\begin{eqnarray*}
\| f\|_{q}^{q} +
\frac{4}{kq(kq-2)}\int\|\nabla|f|^{q/2}|^{2}=\left(\frac{kq-1}{kq-2}\right)\|f\|_{q}^{q}
\end{eqnarray*}
So now our sharp Sobolev inequality yields the following norm
estimate for functions in $\mathfrak{F}_{k}$:
\begin{eqnarray*}
\|f\|_{q}^{q}\geq\left(\frac{2k-1}{kq-1}\right)\left(\frac{kp-1}{2k-1}\right)^{q/p}\|f\|_{p}^{q}
\end{eqnarray*}
\end{proof}

\subsection{A Lower Bound for the Wehrl Entropy of functions in $\mathfrak{F}_{k}$}
We now derive a lower bound for the entropy of functions in
$\mathfrak{F}_{k}$.

\begin{theorem}
The Wehrl entropy associated with
$\mathcal{L}\psi(\zeta)\in\mathfrak{F}_{k}$ has a lower bound given
by:
\begin{eqnarray}\label{3.4}
S(|\mathcal{L}\psi(\zeta)|^{2})\geq
2k\ln\left(1+\frac{1}{2k-1}\right)
\end{eqnarray}
\end{theorem}
\begin{proof}
Let us define, for any function $f$,
$\varphi(p)=\ln||f||_{p}^{p}=\ln\int |f|^{p}$. Then, we have:
\begin{eqnarray*}
S(|f|^{2})=-2\int |f|^{2}\ln|f|=-2\varphi'(2)
\end{eqnarray*}
if $||f||_{2}=1$. By logarithmic convexity of the $p$-norm:
\begin{eqnarray*}
-2\varphi'(2)\geq -2k \varphi\left(2+\frac{1}{k}\right)
\end{eqnarray*}

If we now set $\displaystyle q=2, p=2+\frac{1}{k}$ in Corollary
$3.3$, we have:
\begin{eqnarray*}
\|\mathcal{L}\psi(\zeta)\|_{2+\frac{1}{k}}^{2+\frac{1}{k}}\leq
\left(\frac{2k-1}{2k}\right)
\end{eqnarray*} since $\|\mathcal{L}\psi(\zeta)\|_{2}^{2}=1$, by definition. This implies, in $\mathfrak{F}_{k}$:
\begin{eqnarray*}
\varphi\left(2+\frac{1}{k}\right)\leq
\ln\left(\frac{2k-1}{2k}\right)
\end{eqnarray*}

 Thus:
\begin{eqnarray*}
-2\varphi'(2)&\geq&
-2k\varphi\left(2+\frac{1}{k}\right)\\
\mbox{or, }\quad S(|\mathcal{L}\psi(\zeta)|^{2})&\geq&
2k\ln\left(1+\frac{1}{2k-1}\right)
\end{eqnarray*}
\end{proof}

A comparison between (\ref{2.3}) and (\ref{3.4}) shows that the
estimate obtained above has the conjectured high-spin asymptotics up
to, but not including, first and higher order terms in $(k^{-1})$
because $\displaystyle
\ln\left(1+\frac{1}{2k-1}\right)=2k\left(\frac{1}{2k-1}-\frac{1}{2}\frac{1}{(2k-1)^{2}}+...\right)$.
In fact this is completely analogous to the lower bound Bodmann
[Bod] obtained for coherent state transforms on the sphere
$\mathbb{S}^{2}$.

\section{Entropy-Energy Inequalities on the Hyperbolic Plane
$\mathbb{H}^{2}$} \vspace{.2in}

We say a Riemannian manifold $M$ with measure $d\mathcal{M}$ admits
a logarithmic Sobolev inequality with constant $C$ if:
\begin{eqnarray}\label{sls}
\int_{M} |f|^{2}\ln |f|^{2}d\mathcal{M}\leq C\int_{M}|\nabla
f|^{2}d\mathcal{M}\qquad\mbox{for all $f$ such that
$\displaystyle\int_{M}|f|^{2}d\mathcal{M}=1$}
\end{eqnarray}
Since the Fisher information associated with a function is often
regarded as an ``energy'', one can say that logarithmic Sobolev
inequalities give a bound on the entropy of a function $f$ in terms
of its energy $\displaystyle
E(f)=\int_{M}|\nabla|f||^{2}d\mathcal{M}$.

Even if $C$ is the best possible constant in (\ref{sls}), this is
only one of a whole family of sharp inequalities, and in many
applications, use of the whole family leads to more incisive
results.

To obtain this family of inequalities, one must determine, for  each
$A>0$, the least value of $B$ for which
\begin{eqnarray}\label{3.6}
\int_{M} |f|^{2}\ln |f|^{2}d\mathcal{M}\leq A\int_{M}|\nabla
f|^{2}d\mathcal{M} + B \qquad\mbox{for all $f$ such that
$\displaystyle\int_{M}|f|^{2}d\mathcal{M}=1$}
\end{eqnarray}
is true. Call this optimal choice $B(A)$. If one then defines an
increasing concave function $\Phi$ through
$$\Phi(t) = \inf_{A>0}\{\ At + B(A)\}\ ,$$
one has
$$\int_{M}|f|^{2}\ln|f|^{2}d\mathcal{M}\leq\Phi(E(f))\ $$
for all $f$ with  $\displaystyle\int_{M}|f|^{2}d\mathcal{M}=1$.

Conversely, given the optimal function $\Phi(t)$, $B(A)$ can be
recovered: It is just the $y$--intercept of the tangent line to $y =
\Phi(t)$ at the value of $t$ for which $\Phi'(t) = A$.

Thus, determining an optimal entropy energy inequality is
essentially equivalent to solving an ``$AB$'' type problem in the
sense of Hebey [Heb]: Obviously, if (\ref{3.6}) holds for some $A$
(that is, if, given some $A$, one can find a constant $B$ such that
(\ref{3.6}) is valid), then it holds for all $A'\geq A$. Similarly,
if (\ref{3.6}) is valid for some $B$, it remains valid for all
$B'\geq B$. Thus, it is natural to ask: what is the smallest
constant $A$ (or $B$) for which one can find a constant $B$
(respectively, $A$) such that inequality (\ref{3.6}) holds? In fact,
these questions arise naturally whenever one has a Sobolev-type
inequality on a Riemannian manifold [Heb]. The smallest $A$ for
which (\ref{3.6}) holds is called the first best constant while the
smallest such $B$ is called the second best constant with respect to
the inequality (\ref{3.6}). Given any Sobolev-type inequality on
some Riemannian manifold, Hebey associated two parallel research
programs with the notion of best constants. The $A$-part of the
program gives priority to the first best constant while the $B$-part
is concerned with the second best constant.

As mentioned in the introduction, on $\mathbb{R}^{2}$, the optimal
entropy--energy function $\Phi_{\R^2}(t)$  is given by
$$\Phi_{\R^2}(t) =
\ln\left(\frac{1}{\pi e}t \right)\ .$$Thus:
\begin{eqnarray*}
\int_{\mathbb{R}^{2}} |f|^{2}\ln|f|^{2}\leq \ln\left(\frac{1}{\pi
e}E(f)\right)
\end{eqnarray*}
Equality is achieved when $f$ is an isotropic Gaussian function. For
an appropriate choice of the variance of the Gaussian, the energy
$E(f)$ can take any value, so this inequality is sharp for all
values of $E(f)$.

In the case of $\mathbb{H}^{2}$ , Beckner proved [Bec2] that the
entropy has the same bound as in $\mathbb{R}^{2}$, i.e.,
\begin{eqnarray*}
\int_{\mathbb{H}^{2}}|f|^{2}\ln|f|^{2}\leq\ln\left(\frac{1}{\pi
e}E(f)\right)
\end{eqnarray*}

In other words,
$$\Phi_{\mathbb{H}^{2}} \le \Phi_{\R^2}\ .$$

This result is asymptotically sharp for small $t$ as explained in
the introduction. However, the inequality is actually strict, and
significantly so, for large $t$. Here we prove an improved bound:

For $t>0$, define $\Phi^\star(t) $ by
$$\Phi^\star(t) = \inf_{k\in \N}\left\{
\frac{1}{2}\ln\left[\left(\frac{2k-2}{2k-1}\right)^{2k+1}\left(\frac{2k-1}{2k}\right)^{2k}\left(\frac{2k-1}{4\pi}\right)\left(1+\frac{1}{k(k-1)}t
\right)^{2k+1}\right] \right\}\ .$$ Notice that this is an infimum
over a a family of increasing, concave functions. As such, it is
increasing and concave.

While we cannot explicitly evaluate the infimum that defines
$\Phi^\star(t) $, we have the following result:

\begin{theorem}  For all $t>0$,
$$\Phi_{\mathbb{H}^{2}} \le   \Phi^\star(t)  < \Phi_{\R^2}\ .$$

\end{theorem}

\medskip

\begin{proof}
We start from
 the sharp Sobolev inequality proved in Theorem 3.2, re-written
in terms of the standard measure derived from the Poincare metric.
Recall that the measures $d\mu$ and $d\nu$ are related via:
$\displaystyle d\nu= \frac{2k-1}{4\pi}d\mu$.

If we rescale  $f$ in inequality (\ref{3.1}) so as to make it
$L^{2}$-normalized in the measure $d\mu$ and rewrite the inequality
with respect to $d\mu$, we get:
\begin{eqnarray*}
 \int f^{p}d\mu \leq
 \left(\frac{kq-1}{2k-1}\right)^{p/q}\left(\frac{2k-1}{kp-1}\right)\left(\frac{kq-2}{kq-1}\right)^{p/q}\left(\frac{2k-1}{4\pi}\right)^{p-q/q}\left[\int
 f^{q}d\mu+\frac{4}{kq(kq-2)}\int |\nabla
 f^{q/2}|^{2}d\mu\right]^{p/q}
 \end{eqnarray*}
 Putting $q=2,\quad p=2+1/k$ and using the logarithmic convexity of the $p$-norm as in the proof of
 theorem $3.4$, we obtain the following estimate:

\begin{eqnarray}\label{ne}
\int f^{2}\ln f d\mu \leq
 \frac{1}{2}\ln\left[\left(\frac{2k-2}{2k-1}\right)^{2k+1}\left(\frac{2k-1}{2k}\right)^{2k}\left(\frac{2k-1}{4\pi}\right)\left(1+\frac{1}{k(k-1)}\int
 |\nabla f|^{2}d\mu\right)^{2k+1}\right]
\end{eqnarray}

Since this holds for every $k$, we get an entropy--energy inequality
by taking the infimum over $k$, and this amounts to the inequality
$\Phi_{\mathbb{H}^{2}} \le \Phi^\star(t)$.

It remains to show that $\Phi^\star(t)  < \Phi_{\R^2}$. We shall do
this using the equivalent $A$--$B$ form of the inequality. To make
the tangent line computation and subsequent comparison with
$\Phi_{\R^2}$, and hence Beckner's estimate, we note that,
(\ref{ne}) implies:

\begin{eqnarray}\label{3.9}
\int f^{2}\ln f^{2} d\mu&\leq&
2k\ln\left(\frac{k-1}{k}\right)+\ln\left(\frac{k-1}{2\pi}\right)+\frac{2k+1}{k(k-1)}\int|\nabla
f|^{2}d\mu
\end{eqnarray}
Now Beckner's inequality [Bec2] on the upper half plane is:
\begin{eqnarray}\label{be}
\int|f|^{2}\ln|f|d\mu\leq\frac{1}{2}\ln\left[\frac{1}{\pi e}\int
|\nabla|f||^{2}d\mu\right]\end{eqnarray}

Since the logarithm is a concave function of its argument,
$\displaystyle\frac{\ln x-\ln x_{0}}{x-x_{0}}<\frac{1}{x_{0}}$,
where $x>x_{0}$. If we put $\displaystyle x=\int|\nabla f|^{2}d\mu$
in (\ref{be}), we obtain the following inequality:
\begin{eqnarray}\label{3.10}
\int f^{2}\ln f^{2}d\mu \leq \frac{1}{x_{0}}\int |\nabla
f|^{2}d\mu+\ln x_{0}-\ln\pi-2
\end{eqnarray}

Inequalities (\ref{3.9}) and (\ref{3.10}) have the form
$\displaystyle \int f^{2}\ln f^{2}d\mu\leq
C_{\epsilon}+\epsilon\int|\nabla f|^{2}d\mu$. We would like to see
how the values for the intercept $C_{\epsilon}$ compare for a given
value of the slope $\epsilon$. Let $C_{x_{0}}$ and $C_{k}$ denote
the intercepts for the inequalities parametrized by $x_{0}$ and $k$
respectively. Now, to make the comparison let us put
$\displaystyle\frac{1}{x_{0}}=\frac{2k+1}{k(k-1)}$. Then, for this
value of $x_{0}$ we have:
\begin{eqnarray*}
C_{x_{0}}=\ln x_{0}-\ln\pi-2 
&=&-\left[\frac{1}{2k}+\frac{1}{2}\left(\frac{1}{2k}\right)^{2}+...\right]+\ln(k-1)-\ln
2\pi-2
\end{eqnarray*}
On the other hand:
\begin{eqnarray*}
C_{k}&=&
2k\ln\left(\frac{k-1}{k}\right)+\ln\left(\frac{k-1}{2\pi}\right)\\
&=& \ln(k-1)-\ln
2\pi-2-\frac{1}{k}-\frac{2}{3}\left(\frac{1}{k}\right)^{2}-\frac{1}{2}\left(\frac{1}{k}\right)^{3}-...
\end{eqnarray*}

Thus, for $\displaystyle x_{0}= \frac{k(k-1)}{2k+1}$, we have:
$\displaystyle
C_{x_{0}}-C_{k}=\frac{1}{2k}+\frac{13}{24}\frac{1}{k^{2}}+...$. This
means that the logarithmic Sobolev inequality (\ref{3.9}) actually
gives an improvement on Beckner's inequality (\ref{3.10}) as regards
the second best constant and $\Phi^\star(t)  < \Phi_{\R^2}$.
\end{proof}

Another way to see the extent to which $\Phi^\star$ is a better
estimate of $\Phi_{\mathbb{H}^{2}} $ than is $\Phi_{\R^2}$ is to use
them both to estimate the entropy of our coherent state transforms,
since for them $\displaystyle E(f)=\frac{k}{2}>\frac{k(k-1)}{2k+1}$.

Inserting the value  $\displaystyle E(f)=\frac{k}{2}$ into
$\Phi_{\R^2}$ we obtain, using Beckner's estimate with respect to
the measure $d\nu(\zeta)$:
\begin{eqnarray*}
-\int |f|^{2}\ln|f|^{2} d\nu\geq 1-\ln\left(\frac{2k}{2k-1}\right)
\end{eqnarray*}
while inserting this value into $\Phi^\star$ (with respect to
measure $d\nu(\zeta)$) yields the better bound
 (\ref{3.4}).

We close this section by proving another family of logarithmic
Sobolev inequalities on the hyperbolic plane. The basic idea comes
from Beckner's paper [Bec1] where he showed how one could derive a
family of sharp Sobolev inequalities on the hyperbolic plane
$\mathbb{H}^{2}$,  from the sharp Sobolev inequality on
$\mathbb{R}^{n}$, for $n>2$.

The sharp Sobolev inequality on $\mathbb{R}^{n}$, for $n>2$ and
$1/p=1/2-1/n$  is given by [Bec1]:
\begin{eqnarray}\label{3.11}
||f||_{L^{p}(\mathbb{R}^{n})}\leq A_{p}||\nabla
f||_{L^{2}(\mathbb{R}^{n})}\nonumber \\A_{p}=[\pi
n(n-2)]^{-1/2}[\Gamma(n)/\Gamma(n/2)]^{1/n}
\end{eqnarray}
and the sharp constant is attained only for functions having the
form $A(1+|\mathbf{x}|^{2})^{-n/p}$, where
$\mathbf{x}\in\mathbb{R}^{n}$.
\begin{theorem}
The sharp Sobolev inequality (\ref{3.11}) on $\mathbb{R}^{n}$ leads
to the following one-parameter family of logarithmic Sobolev
inequalities on the hyperbolic plane $\mathbb{H}^{2}$:

\begin{eqnarray}\label{3.12}
\int g^{2}\ln g^{2}d\mu\leq
\tilde{k}\ln\left[\left(\frac{\tilde{k}-1}{\tilde{k}+1}\right)^{1+1/\tilde{k}}\left(\frac{2\tilde{k}+1}{2\pi}\right)^{1/\tilde{k}}\left(1+\frac{1}{\tilde{k}(\tilde{k}-1)}\int
|Dg|^{2}d\mu\right)^{1+1/\tilde{k}}\right]
\end{eqnarray}
where $\tilde{k}=n/p$.
\end{theorem}

\begin{proof}
To obtain (\ref{3.12}), we first derive a family of sharp Sobolev
inequalities on the hyperbolic plane $\mathbb{H}^{2}$, as mentioned
in [Bec1]. In order to do this, let us restrict our computations to
radial functions $f$ in inequality (\ref{3.11}). We don't lose
anything by doing this since the optimizer is radial. Let us use the
product structure for Euclidean space $\mathbb{R}^{n}\simeq
\mathbb{R}\times\mathbb{R}^{n-1}$, with
$\mathbf{x}\in\mathbb{R}^{n}$ written as $(t,\mathbf{x'})$ where
$\mathbf{x'}\in\mathbb{R}^{n-1}$. Also let $y=|\mathbf{x'}|$. Now
put $g(t,y)=y^{n/p}f(t,\mathbf{x'})$. Then:
\begin{eqnarray*}
 \int f^{p} d\mathbf{x} &=& \int y^{-n}g^{p}dtd\mathbf{x'}\\
 &=& \int y^{-n}g^{p}dt (y^{n-2}dy S^{n-1})\\
 &=& S^{n-1}\int g^{p}d\mu
 \end{eqnarray*}
Here $\displaystyle d\mu= dt\frac{dy}{y^{2}}$ is the measure derived
from the Poincare metric on the upper half plane (recall that it is
equivalent to the hyperbolic plane) and in moving from the first to
the second line in the computation above we have referred
 $\mathbf{x'}$ to the $(n-1)$-dimensional spherical polar coordinate
system, so that $S^{n-1}$ in the final expression represents the
surface area of a $(n-1)$-dimensional sphere. After a similar
computation and some simplification, the expression for the
gradient-norm of $f$ is obtained as:
\begin{eqnarray*}
\int |\nabla f|^{2} d\mathbf{x} = S^{n-1}\left[\int |Dg|^{2}d\mu
+\frac{n}{p}\left(\frac{n}{p}-1\right)\int g^{2}d\mu\right]
\end{eqnarray*}
Thus inequality (\ref{3.11}) can be expressed as:
\begin{eqnarray*}
\left(\int_{\mathbb{H}^{2}} g^{p}d\mu\right)^{2/p}\leq
(S^{n-1})^{1-2/p}A_{p}^{2}\left(\int_{\mathbb{H}^{2}} |Dg|^{2}d\mu
+\frac{n}{p}\left(\frac{n}{p}-1\right)\int_{\mathbb{H}^{2}}
g^{2}d\mu\right)
\end{eqnarray*}
A short computation now shows:
$\displaystyle(S^{n-1})^{1-2/p}A_{p}^{2}=\frac{4}{n(n-2)}\left(\frac{n-1}{2\pi}\right)^{2/n}$.
We can rewrite the inequality on the hyperbolic plane as:
\begin{eqnarray*}
\int_{\mathbb{H}^{2}}g^{p}d\mu\leq
\left(\frac{4}{n(n-2)}\right)^{p/2}\left(\frac{n-1}{2\pi}\right)^{p/n}\left(\frac{n}{p}\left(\frac{n}{p}-1\right)\right)^{p/2}\left[\int_{\mathbb{H}^{2}}g^{2}d\mu
+ \frac{1}{\frac{n}{p}(\frac{n}{p}-1)}\int |Dg|^{2}d\mu\right]^{p/2}
\end{eqnarray*}
Since $p=2n/(n-2)$, the inequality above represents a one-parameter
family of inequalities. Let us introduce  a new variable
$\tilde{k}=(n-2)/2$. Then $\tilde{k}= 1/2, 1, 3/2, 2, 5/2,...$ and
$p=2+2/\tilde{k}$. In terms of $\tilde{k}$ we have the following
family of inequalities:
\begin{eqnarray}\label{3.13}
\int_{\mathbb{H}^{2}} g^{p}d\mu\leq
\left(\frac{\tilde{k}-1}{\tilde{k}+1}\right)^{1+1/\tilde{k}}\left(\frac{2\tilde{k}+1}{2\pi}\right)^{1/\tilde{k}}\left[\int_{\mathbb{H}^{2}}
g^{2}d\mu+\frac{1}{\tilde{k}(\tilde{k}-1)}\int_{\mathbb{H}^{2}}
|Dg|^{2}d\mu\right]^{1+1/\tilde{k}}
\end{eqnarray}
The sharp constant in this inequality is attained for functions
$\displaystyle
g^{*}=A\left(\frac{y}{1+t^{2}+y^{2}}\right)^{n/p}=A\left(\frac{y}{1+t^{2}+y^{2}}\right)^{\tilde{k}}$.

Using the logarithmic convexity of the $p$-norm we obtain from
(\ref{3.13}), the family of logarithmic Sobolev inequalities
(\ref{3.12}) for functions on the hyperbolic plane, which are
normalized so that their $L^{2}$-norm with respect to the measure
$d\mu$ is $1$ :
\begin{eqnarray*}
\int g^{2}\ln g^{2}d\mu\leq
\tilde{k}\ln\left[\left(\frac{\tilde{k}-1}{\tilde{k}+1}\right)^{1+1/\tilde{k}}\left(\frac{2\tilde{k}+1}{2\pi}\right)^{1/\tilde{k}}\left(1+\frac{1}{\tilde{k}(\tilde{k}-1)}\int
|Dg|^{2}d\mu\right)^{1+1/\tilde{k}}\right]
\end{eqnarray*}
\end{proof}

It is interesting to note that, for $k=3/2,2,5/2...$, one can obtain
from Theorem $3.2$, a one-parameter family of Sobolev inequalities,
which is strikingly similar to (\ref{3.13}). Referred to the
standard measure $d\mu$, Theorem $3.2$ tells us, for $q=2$ and
$p'=2+1/k$:
\begin{eqnarray}\label{3.14}
\int_{\mathbb{H}^{2}} f^{p'}d\mu\leq
\left(\frac{2k-2}{2k}\right)^{1+1/2k}\left(\frac{k}{2\pi}\right)^{1/2k}\left[\int_{\mathbb{H}^{2}}
f^{2}d\mu +\frac{1}{k(k-1)}\int_{\mathbb{H}^{2}}
|Dg|^{2}d\mu\right]^{1+1/2k}
\end{eqnarray}

It is thus very natural to compare (\ref{3.12}) with (\ref{3.4}) and
see which inequality gives a better bound for the entropy of
functions in $\mathfrak{F}_{k}$. Let us first see what (\ref{3.12})
implies for such functions. Put $\displaystyle
g=\sqrt{\frac{2k-1}{4\pi}}f$ where $f\in\mathfrak{F}_{k}$. Then we
have: $\int g^{2}d\mu=1$ and $\int |Dg|^{2}d\mu=k/2$. So, with
reference to the coherent state measure $\displaystyle
d\nu=\frac{2k-1}{4\pi}d\mu$, (\ref{3.12}) implies that:
\begin{eqnarray*}
-\int f^{2}\ln f^{2}d\nu\geq
\ln\left(\frac{2k-1}{4\pi}\right)-(1+\tilde{k})\ln\left(\frac{\tilde{k}-1}{\tilde{k}+1}\right)-\ln\left(\frac{2\tilde{k}+1}{2\pi}\right)-(1+\tilde{k})\ln\left(1+\frac{k}{2\tilde{k}(\tilde{k}-1)}\right)
\end{eqnarray*}
Optimization of the right hand side over the parameter $\tilde{k}$
doesn't seem to yield a simple result. However, we can put
$\tilde{k}=2k$ (where $k=3/2,2,...$) to make $p=p'$, so that we have
the same $L^{p}$-norms on the left hand sides of (\ref{3.13}) and
(\ref{3.14}). The resulting expression yields the lower bound:
\begin{eqnarray*}
-\int |f|^{2}\ln|f|^{2}d\nu\geq
\ln\left(\frac{2k-1}{2(4k+1)}\right)-(1+2k)\ln\left(\frac{4(2k-1)}{2k+1}\right)
\end{eqnarray*}
Obviously, Theorem $3.4$ gives a better bound for the entropy of
functions in $\mathfrak{F}_{k}$.

\section{The uniqueness theorem} \vspace{.2in}

 In this section we study (3.2) written in terms of the radial hyperbolic
coordinate. 
In what follows, we
adapt the methods described in [Kwo] to the hyperbolic setting. 

 We investigate the question of uniqueness of ground state
 solution of the equation

\begin{eqnarray}
 \displaystyle u'' +\coth\tau u'+f(u)=0
\end{eqnarray}
 where $\tau\in(0,\infty)$ on the  two-dimensional hyperbolic plane. The function $f(u)$ is given by: $\displaystyle f(u)= \tilde{a}u^{1+\frac{2}{kq}}-\tilde{b}u$, where 
 $\displaystyle\tilde{b}=kq(kq-2)$ and $\displaystyle\tilde{a}=\gamma kq(kq-2)$. The
 boundary conditions on the solutions of interest are: $\lim_{\tau\longrightarrow\infty}u(\tau)=0 $ and
 $u'(0)=0$.
There exist three points $\xi_{0}$, $\xi_{1}$ and $\xi_{2}$ in
$(0,\infty)$ such that:
\begin{eqnarray*}
 \int_{u=0}^{\xi_{0}}f(u)du&=& 0\mbox{;\hspace{.5in}}\int_{u=0}^{v}f(u)du< 0 \mbox{ for $v<\xi_{0}\quad$ and } \int_{u=0}^{v}f(u)du> 0 \mbox{ for $v>\xi_{0}$}\\
\mbox{\vspace{2in}}\\
 f(\xi_{1})&=& 0\mbox{;\hspace{.5in}} f(u)<0 \quad \mbox{  if $\quad u<\xi_{1}\quad$  and  }\quad f(u)>0 \mbox{   if  $u>\xi_{1}$}\\
 f'(\xi_{2})&=&0 \mbox{;\hspace{.5in}} f'(u)<0 \quad \mbox{  if $\quad u<\xi_{2}\quad$  and  }\quad f'(u)>0 \mbox{   if  $u>\xi_{2}$}\\
\end{eqnarray*}
\begin{figure}[hbt]\centering
 \begin{tabular}{c}
 \includegraphics[angle=270,width=10cm]{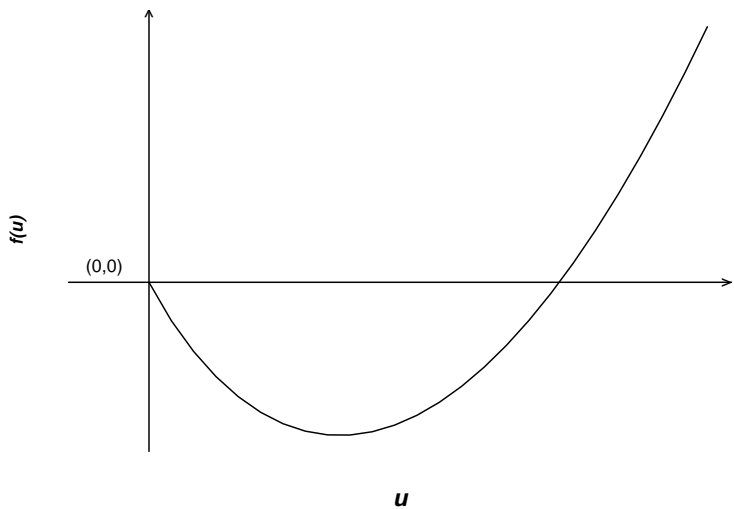}\\[-5pt]
 \end{tabular}
 \caption{The function $f(u)$}
 \end{figure}

Following [Kwo], let us consider $u$ as a function of the initial
value $\alpha$ and $\tau$, and study, in stead of the boundary value
problem mentioned above, the following initial value problem:
\begin{eqnarray}\label{3.15}
u'' +\coth\tau u'+f(u)=0\\
u(0)=\alpha>0, \qquad u'(0)=0\nonumber
\end{eqnarray}

We first divide the set of solutions into three mutually disjoint
subsets, namely:

\begin{enumerate}

\item Solutions that have a zero at some finite $\tau$. We call the corresponding set of
initial values $N$. We denote the finite zero as $b(\alpha)$.

\item Positive solutions that satisfy
$\lim_{\tau\rightarrow\infty}u(\tau)=0$. We call the set of initial
values $G$ in this case.

\item Solutions that remain positive and do not belong to case 2.
We let $P$ denote the set of initial values for such solutions.

\end{enumerate}

For a particular solution $u\in G\cup N$, we let $\tau_{1}$ denote
the zero of $f(u)$, that is to say, $u(\tau_{1})=\xi_{1}$ (it is
possible to define this point uniquely because, as we will show
momentarily, solutions $u\in G\cup N$ are monotone). Our subsequent
results rely heavily on Sturm's comparison theorem (as mentioned in
[lemma 1, [Kwo]] and also in chapter X, page 229 of [Inc]) and a few
important corollaries that we state below. 

Consider two second order differential equations:
\begin{eqnarray}\label{3.16}
U''(x)+f(x)U'(x)+g(x)U(x)=0, \qquad x\in(a,b)
\end{eqnarray}
\begin{eqnarray}\label{3.17}
V''(x)+f(x)V'(x)+G(x)V(x)=0, \qquad x\in(a,b)
\end{eqnarray}

Suppose that (\ref{3.16}) has solutions that do not vanish in a
neighborhood of point $b$. Then the largest neighborhood of $b$,
$(c, b)$, on which there exists a solution of (\ref{3.16}) without
any zero, is called the disconjugacy interval of (\ref{3.16}).
Sturm's theorem implies that no non-trivial solution can have more
than one zero in $(c, b)$. A corollary (lemma 6, [Kwo]) of Sturm's
theorem is: if $(c, \infty)$ is the discongugacy interval of
(\ref{3.16}), as defined above, then every solution of (\ref{3.16})
with a zero in $(c, \infty)$ is unboounded. We also have another
very useful corollary (lemma 3, [Kwo]) of Sturm's theorem: if the
equations (\ref{3.16}) and (\ref{3.17}) satisfy the comparison
condition $G(x)\geq g(x)$, $U$ is not identically equal to $V$ in
any neighborhood of $b$ and there exists a solution $V$ of
(\ref{3.17}) with a largest zero at $\rho\in(a,b)$, then the
disconjugacy interval of (\ref{3.16}) is a strict superset of
$(\rho,b)$.

We are now ready to state and prove our results. But first let us
briefly outline our strategy in a few steps, since the proof of
uniqueness is rather involved:
\begin{enumerate}
\item The first two lemmas state well-known facts about the
structure of the sets $N$, $P$ and $G$. As we increase $\alpha$ from
$0$ we first have solutions in $P$. Since the arguments are exactly
similar to those used for the Euclidean case in [Kwo], we refer to
the relevant lemmas in [Kwo], in stead of reiterating the proofs.

\item Next we study the variation $w$ of a solution $u\in G\cup N$ with
respect to its initial value. 
The proof of uniqueness depends crucially on the properties of $w$.
If, for $\alpha\in G$
$\lim_{\tau\longrightarrow\infty}w(\alpha,\tau)=-\infty$, then a
right neighborhood of $\alpha$ belongs to $N$. Also, if $\alpha\in
N$ and $w(\alpha,b(\alpha))<0$, then a neighborhood of $\alpha$
belongs to $N$ as well. Suppose these hypotheses are indeed true. As
we continuously increase $\alpha$, we will first have solutions in
$P$. The right boundary point will belong to $G$. 
A right neighborhood of the corresponding $\alpha$ will be in $N$.
Then, if for all $\alpha\in N$, $w(\alpha,b(\alpha))<0$, we will
continue to remain in $N$ as we increase $\alpha$ further. Thus the
proof of uniqueness of the ground state will be complete. Hence we
just need to prove that for $\alpha\in G$,
$\lim_{\tau\longrightarrow\infty}w(\alpha,\tau)=-\infty$, while for
$\alpha\in N$,  $w(\alpha,b(\alpha))<0$. In fact, if we can prove
that $w$ has only one zero for initial values in $G\cup N$ and $w$
is unbounded for initial values in $G$, uniqueness will be
guaranteed. Initial values satisfying these two conditions are
called strict admissible.

\item To prove that $w$ can have no more than one zero and that it
is unbounded, we construct a comparison function $v$ for $w$
The zero of $w$ is then shown to belong to the disconjugacy interval
of the differential equation satisfied by $w$, which in turn implies
unboundedness of $w$. The idea of constructing a comparison function
like this was used in [Kwo] to prove uniqueness of positive
solutions of a semi-linear Poisson equation in a bounded or
unbounded annular region in $\mathbb{R}^{n}$, for $n>1$. It is in
this crucial step, right after lemma $5.5$ in this paper, that our
proof of uniqueness differs from that of [Kwo]. This happens because
we are dealing with a semi-linear Poisson equation on the hyperbolic
plane $\mathbb{H}^{2}$. The difference in geometry manifests itself
in the form of the comparison function and, more importantly, in the
subsequent analysis. Proofs of lemma $5.6$ through lemma $5.8$ are
thus specific to the hyperbolic case. As we go along we point out
these differences  in detail.

\end{enumerate}

The main result of this section is:
\begin{theorem}
The initial value $\alpha\in G\cup N$ is strictly admissible.
\end{theorem}
\vspace{.2in}

Let us construct an ``energy" function corresponding to
(\ref{3.15}):
\begin{eqnarray*}
E(\tau) =
\frac{u'^{2}(\tau)}{2}+\frac{\tilde{a}u^{2+\frac{2}{kq}}}{2+\frac{2}{kq}}-\frac{\tilde{b}u^{2}}{2}
\end{eqnarray*}

It is readily seen that $\displaystyle E'(\tau)=-\coth\tau
u'^{2}(\tau)\leq 0$. Thus $E$ is a non-increasing function of
$\tau$.

\begin{lemma}
The set $(0,\xi_{0}]$ of initial values belongs to the set P. [lemma
8, [Kwo]]
\end{lemma}

 For solutions in $N$, the function $E$ decreases to a positive
constant while for solutions in $G$, $E(\infty)=0$. This fact leads
us to the following lemma:
\begin{lemma}

If $u\in G\cup N$, then $u'(\tau)<0$ in $(0,b(\alpha))$ (if $u\in
b(\alpha)$) or $(0,\infty)$ (if $u\in G$). [lemma 11, [Kwo]]
\end{lemma}

The fact that the sets $N$ and $P$ are open subsets of $(0,\infty)$
[lemma 13, [Kwo]; lemma 1.1, [Ber]] is crucial but easy to observe.

We concern ourselves only with solutions that are either in $G$  or
in $N$ . Let us define: $\displaystyle
w=w(\tau,\alpha)=\left.\frac{\partial
u}{\partial\alpha}\right|_{\tau,\alpha}$. We study the function $w$
for such solutions. First of all let us note that $w=0$ means two
nearby solutions (i.e. solutions having nearby initial values) can
intersect.

      Evidently $w$ satisfies the following equation (the derivatives are taken with respect to $\tau$):
\begin{eqnarray}\label{3.18}
w''+\coth\tau w'+f'(u)w=0\\
w(0)=1, \qquad w'(0)=0\nonumber
\end{eqnarray}
\begin{lemma}
For $u\in G\cup N$, $w$ has to change sign before $\xi_{1}$. [lemma
17, [Kwo]]
\end{lemma}

Following Kwong, we call the initial value $\alpha\in G$ strictly
admissible if the corresponding $w(\alpha,\tau)$ has only one zero
in $(0,\infty)$ and
$\lim_{\tau\longrightarrow\infty}w(\alpha,\tau)=-\infty$. We call
the initial value $\alpha\in N$ strictly admissible if the
corresponding $w(\alpha,\tau)$ has only one zero in $(0,\infty)$ and
$w(\alpha,b(\alpha))<0$.

It is easy to see that if for a particular $\alpha\in N$,
$\displaystyle w(b(\alpha))=\frac{\partial
u}{\partial\alpha}(b(\alpha),\alpha)<0$, then in a right
neighborhood of $\alpha$, $b(\alpha)$ is a strictly decreasing
function of $\alpha$ and thus that neighborhood belongs to $N$.

\begin{lemma}
If for $\alpha\in G$,
$\lim_{\tau\longrightarrow\infty}w(\alpha,\tau)=-\infty$, in
particular if $w(\alpha,\tau)$ is strictly admissible, then there
exists a right neighborhood of $\alpha$ that belongs to $N$. [lemma
19, [Kwo]]
\end{lemma}

We now need to prove that every initial value $\alpha\in G\cup N$ is
strictly admissible. The strategy is to construct a comparison
function $v(\tau)$ (to be compared with $w$), which has the
following properties:
\begin{enumerate}
\item $v(\tau)$ has only one zero in $(0,\infty)$.

 and

 \item $v(\tau)$ is a strict Sturm majorant
of $w(\alpha,\tau)$ in both $(0,\rho)$ and $(\rho,\infty)$, where
$\rho$ is the first zero of $w(\alpha,\tau)$.
\end{enumerate}

If we are able to construct such a function, then by property (2)
the zero of $v$ occurs before that of $w$ and by property (1) $w$
cannot have another zero in $(0,b(\alpha))$. Here $b(\alpha)$ is the
zero of the solution $u\in G\cup N$. If $u\in G$ then $b(\alpha)$ is
to be interpreted as the point $\tau=\infty$. If $b(\alpha)$ is
finite then of course the corresponding $u$ is in $N$ and
$w(\alpha,b(\alpha))<0$, i.e., $\alpha$ is strictly admissible. On
the other hand if $b(\alpha)=\infty$, $w$ has a zero in the
disconjugacy interval of $v$, and hence in the disconjugacy interval
of the differential equation satisfied by $w$ itself. This happens
because $w$ being a strict Sturm minorant of $v$ in $(0,\infty)$,
the disconjugacy interval of (\ref{3.16}) is bigger than that of the
differential equation satisfied by $v$. This means $w$ is unbounded.
Hence the corresponding $\alpha$ is strictly admissible.

It is helpful to first construct an auxiliary function
$\theta(\tau)$ and then use it to deduce that $v$ has the necessary
properties described above. In the Euclidean case [Kwo], the
auxiliary function $\theta(r)$ is given by:
$\displaystyle\theta(r)=-\frac{ru'(r)}{u(r)}$. For the hyperbolic
case we define the auxiliary function
 for all solutions
$u\in G\cup N$ as:
\begin{eqnarray}\label{3.21}
\theta(\tau)=\frac{-\sinh\tau u'(\tau)}{u(\tau)}
\end{eqnarray}
The auxiliary functions and the comparison functions in the
Euclidean and hyperbolic cases have different forms but similar
properties. Thus lemmas that follow are basically hyperbolic
analogues of lemmas proved by Kwong in the Euclidean case.

The function $\theta(\tau)$ is obviously continuous in $(0,\infty)$
for $u\in G$; for $u\in  N$ $\theta(\tau)$ is continuous in
$(0,b(\alpha))$ where $b(\alpha)$ is the zero of $u(\alpha)$.

\begin{lemma}
For solutions $u\in G\cup N$, $\theta(0)=0$ and
$\lim_{\tau\longrightarrow b(\alpha)}\theta(\tau)=\infty$. If $u\in
N$ $b(\alpha)$ is interpreted to be the zero of $u$ and if $u\in G$,
$b(\alpha)=\infty$.
\end{lemma}
\begin{proof}

The first claim is easy to verify since for all $u\in G\cup N$,
$u'(0)= 0$; since $u'(\tau)<0$, $\theta(\tau)> 0$ in $(0,\infty)$.

For $u\in N$, $u'(b(\alpha))\neq 0$ and the second assertion of the
lemma automatically follows.

Let us consider the case: $u\in G$.

Let $\displaystyle R=-\frac{u'}{u}$.

Then $R\geq 0$ and $\displaystyle
R'=-\frac{u''}{u}+\frac{u'^{2}}{u^{2}}=R^{2}-R\coth\tau
+\frac{f(u)}{u}$.

Now we know that $\displaystyle
\lim_{\tau\longrightarrow\infty}\frac{f(u)}{u}=-\tilde{b}$. We
assert that for large values of $\tau$ we would always have:
$\displaystyle R(\tau)> \sqrt{\frac{\tilde{b}}{2}}$. If not, then
$\displaystyle R(\tau)\leq \sqrt{\frac{\tilde{b}}{2}}$ for some
$\tau$. Then:

$\displaystyle R'(\tau)=R^{2}-\coth\tau
R+\frac{f(u)}{u}<R^{2}+\frac{f(u)}{u}\leq-\frac{\tilde{b}}{2}$.

Thus $R'$ will remain strictly and hugely negative eventually
causing $R$ to change sign.

Thus
$\displaystyle-\frac{u'(\tau)}{u(\tau)}>\sqrt{\frac{\tilde{b}}{2}}$
for large values of $\tau$. This in turn means
$\lim_{\tau\longrightarrow\infty}\theta(\tau)=\infty$.
\end{proof}
\vspace{.15in}

We next define the comparison function $v_{\beta}(\tau)=\sinh\tau u'
+\beta u$ (in the Euclidean case it is defined as
$v_{\beta}(r)=ru'(r)+\beta u(r)$) . It is readily seen that
$v_{\beta}(\tau)=(<,>)0$ if and only if $\theta$ intersects (is
above, is below) the straight line $y(\tau)=\beta$. Also,
$v_{\beta}(\tau)$ is tangent to the $\tau$-axis at some point
$\hat{\tau}$ if and only if $\theta(\tau)$ is tangent to the
straight line $y(\tau)=\beta$ at $\hat{\tau}$.

The function $v_{\beta}(\tau)$ satisfies the following differential
equation:

\begin{eqnarray}\label{3.22}
 v''+\coth\tau v'+f'(u)v=\Phi(\tau)=\beta(u f'(u)-f(u))-2\cosh\tau f(u)\\
 v(0)>0, \quad v'(0)= 0\mbox{\hspace{2in}}\nonumber
\end{eqnarray}

Now \begin{eqnarray*}
 \Phi &=&\beta\left(uf'(u)-f(u)\right)-2\cosh\tau f(u)\\
 &=&\frac{2}{kq}\beta\tilde{a}u^{1+\frac{2}{kq}}-2\cosh\tau f(u)
 \end{eqnarray*}

It is not really obvious that one can choose a $\beta$ such that
$\Phi$ has only one zero and the position of that zero has a
continuous dependence on $\beta$. However our next lemma proves that
this can indeed be achieved.
\begin{lemma}
There exists some $\bar{\beta}$ such that for $0<\beta<\bar{\beta}$
the function $\Phi(u,\tau)$ has only one zero, say at $\tau=\sigma$
in $(0,\infty)$ such that:
\begin{eqnarray*}
\Phi(u,\tau)&<& 0\qquad\mbox{for }\tau<\sigma\\
\Phi(u,\tau)&>& 0\qquad\mbox{for }\tau>\sigma
\end{eqnarray*}

The point $\sigma$ is a continuous monotone function of $\beta$.
\end{lemma}
\begin{proof}
First, we note that $\Phi(\tau)>0$ in $[\tau_{1},\infty)$ by
definition; so its zeros must be concentrated in $(0,\tau_{1})$. At
a zero of the function $\Phi$ we have:
\begin{eqnarray*}
 \frac{2}{kq}\beta\tilde{a}u^{1+\frac{2}{kq}}&&=2\cosh\tau f(u)
\end{eqnarray*}
 Thus at $\Phi=0$ we have:
\begin{eqnarray*}
 \Phi'&=&\beta\frac{2}{kq}\tilde{a}\left(1+\frac{2}{kq}\right)u^{\frac{2}{kq}}u'-2\sinh\tau f(u)-2\cosh\tau f'(u)u'\\
 &=&\frac{2u'\cosh\tau}{u}\left[\left(1+\frac{2}{kq}\right)f(u)-uf'(u)\right]-2\sinh\tau f(u)
 \end{eqnarray*}
So, if at $\Phi=0$, $\Phi'>0$, then:
\begin{eqnarray*}
 \frac{2u'\cosh\tau}{u}\left[\left(1+\frac{2}{kq}\right)f(u)-uf'(u)\right]> 2\sinh\tau f(u)\\
\mbox{or ,}\quad
-\frac{2\tilde{b}}{kq}u'>\tanh\tau f(u)
\end{eqnarray*}
which in turn implies
\begin{eqnarray}\label{3.23}
\frac{2\tilde{b}}{kq}(-\sinh\tau u')>\sinh\tau\tanh\tau f(u)
\end{eqnarray}
Similarly if $\Phi'<0$ at $\Phi=0$, then:
\begin{eqnarray}\label{3.24}
\frac{2\tilde{b}}{kq}(-\sinh\tau u')<\sinh\tau\tanh\tau f(u)
\end{eqnarray}

Now the differential equation (\ref{3.15}) satisfied by $u$ can be
rewritten as:
\begin{eqnarray*}
(-\sinh\tau u')'=\sinh\tau f(u)
\end{eqnarray*}
If at the first zero of the function $\Phi(\tau)$, $\Phi'(\tau)>0$
then inequality (\ref{3.23}) holds at that point and we also know
that the left hand side of the inequality is positive and increasing
at the rate $\displaystyle\left(\frac{2\tilde{b}}{kq}(-\sinh\tau
u')\right)'= \frac{2\tilde{b}}{kq}\sinh\tau f(u)$. As for the right
hand side, we have, in the interval $(0,\tau_{1})$:
\begin{eqnarray*}
(\sinh\tau\tanh\tau f(u))'&=&\sinh\tau
f(u)+\sinh\tau\mbox{sech}^{2}\tau f(u)+\sinh\tau\tanh\tau
f'(u)u'\\&<& 2\sinh\tau f(u)
\end{eqnarray*}
The inequality above holds because $f'(u)>0$ in $(0,\tau_{1})$ and
$u'<0$. Since in our case
$\displaystyle\frac{2\tilde{b}}{kq}=2(kq-2)$ and $k$ is chosen so
that $kq> 1$, it turns out that
$\displaystyle\frac{2\tilde{b}}{kq}\sinh\tau f(u)> 2\sinh\tau f(u)$.
This in turn implies that the left hand side of (\ref{3.23})
increases more rapidly than the right hand side. So if inequality
(\ref{3.23}) holds at some point in $(0,\tau_{1})$ then it prevails
at all subsequent points in this interval. We can thus conclude that
if $\Phi(0)<0$, then $\Phi(\tau)$ can have only one zero in
$(0,\tau_{1})$.

Now for a particular solution having initial value $\alpha$
,$\Phi(\tau=0)=\beta\left(\alpha
f'(\alpha)-f(\alpha)\right)-2f(\alpha)$. Putting in the specific
form of $f(u)$ we obtain the condition that $\Phi(\tau)$ has a
negative initial value:
\begin{eqnarray*}
\beta< kq\left[1-\frac{\tilde{b}}{\tilde{a}\alpha^{2/kq}}\right]
\end{eqnarray*}
We let $\bar{\beta}$ denote the upper limit set on $\beta$ by the
condition above. Then for $\beta\in(0,\bar{\beta})$, the function
$\Phi(\tau)$ has a negative initial value and consequently only one
zero in $(0,\tau_{1})$. We denote that zero by $\sigma$.

Let us now find out how $\sigma$ depends on $\beta$. We have:
\begin{eqnarray*}
\frac{\tilde{a}\beta}{kq}=
\cosh\sigma\frac{f(u(\sigma))}{u(\sigma)^{1+2/kq}}
\end{eqnarray*}
Evidently then $\beta$ depends continuously on $\sigma$. Also:
\begin{eqnarray*}
\beta'(\sigma)
&=&\frac{kq}{\tilde{a}}u^{-1-2/kq}\left[\frac{2\tilde{b}}{kq}u'(\sigma)\cosh\sigma+f(u)\sinh\sigma\right]
\end{eqnarray*}
Now for $\beta\in(0,\bar{\beta})$, (\ref{3.23}) holds at $\sigma$,
as proved before. Thus $\displaystyle
\left[\frac{2\tilde{b}}{kq}u'(\sigma)\cosh\sigma+f(u)\sinh\sigma\right]<0$,
and hence $\displaystyle \beta'(\sigma)<0$ for all $\beta$ in this
range. This means there exists a continuous inverse function in a
neighborhood of $\beta(\sigma)$. Thus $\sigma$ depends continuously
on $\beta$. In fact $\sigma$ is a decreasing function of $\beta$.
When $\beta=0$ the only zero of $\Phi(\tau)$ is at $\tau_{1}$. As we
increase $\beta$ the zero shifts continuously to the left.
\end{proof}
Let $\rho_{\beta}$ be the first zero of $v_{\beta}(\tau)$ (we do not
yet know how many zeros $v$ can have). Then for $\beta=0$, $\rho=0$.
As we increase $\beta$, $\rho_{\beta}$ moves to the right. In order
to prove that we can control $\beta$ such that $\rho_{\beta}$ and
$\sigma_{\beta}$ can be made to coincide, we need to show that
$\rho_{\beta}$ continuously depends on $\beta$. We first show that
actually, given any $\beta$, $v_{\beta}(\tau)$ can have only one
zero and then prove the continuous dependence of that zero on the
parameter $\beta$.

\begin{lemma}
The function $v_{\beta}(\tau)$ has only one zero in $(0,\infty)$.
\end{lemma}
\begin{proof}
In the interval $[0,\tau_{1}]$,
\begin{eqnarray*}
(-\sinh\tau u'(\tau))'=f(u)\sinh\tau\geq 0
\end{eqnarray*}
Thus $(-\sinh\tau u'(\tau))$ is non-decreasing in $[0,\tau_{1}]$.
Since $u(\tau)$ is decreasing,
$\displaystyle\theta(\tau)=\frac{-\sinh\tau u'(\tau)}{u(\tau)}$ is
non-decreasing in $[0,\tau_{1}]$. Thus for any $\beta$ it can
intersect the straight line $y(\tau)=\beta$ no more than once in
this interval and  the corresponding $v_{\beta}(\tau)$ can have at
most one zero.

Since $\lim_{\tau\longrightarrow\infty}\theta(\tau)=\infty$, if
$\theta(\tau)$ is not non-decreasing in the entire interval
$(\tau_{1},\infty)$, then it has to have local minima. Suppose the
lowest of all such minima occurs  at $\omega$ and has height
$\beta_{0}$. Then in $(\omega,\infty)$, $v_{\beta_{0}}(\tau)$ is
negative and has a double zero at $\omega$. Also
$v_{\beta_{0}}(\tau)$ satisfies the following differential
inequality in $(\omega,\infty)$:
\begin{eqnarray*}
v''+\coth\tau v' +f'(u)v\geq 0
\end{eqnarray*}
But this is impossible (since, if $v$ satisfies the second-order
differential equation above, then it cannot have a double zero; cf.
lemma 5, [Kwo]).

Thus we conclude that $\theta(\tau)$ is non-decreasing in
$(0,\infty)$, which in turn implies that for any value of $\beta$,
$v_{\beta}(\tau)$ can have only one zero in $(0,\infty)$.
\end{proof}

To prove that one can choose $\beta$ such that
$\rho_{\beta}=\sigma_{\beta}$ it is sufficient to show that
$\rho_{\beta}$ as a function of $\beta$ doesn't have any
discontinuity in $(0,\tau_{1})$. Since $v_{\beta}$ has a zero at
$\rho_{\beta}$ if and only if $\theta$ intersects the straight line
$y(\tau)=\beta$ at $\tau=\rho_{\beta}$, we just need to show
$\theta'(\rho_{\beta})\neq 0$. As shown in the preceding lemma,
$\theta'(\tau)>0$ in $(0,\tau_{1})$. As we increase $\beta$, the
height of the horizontal straight line $y(\tau)=\beta$ increases.
This results in a continuous shift of the point of intersection
$\rho_{\beta}$ to the right. Thus we can conclude that in
$(0,\tau_{1})$ $\rho_{\beta}$ is a continuous increasing function of
$\beta$. For $\beta=0$, $\rho=0$ and $\sigma=\tau_{1}$. When we
increase $\beta$, $\rho_{\beta}$ moves continuously to the right
even as $\sigma_{\beta}$ shifts continuously to the left till it is
at the origin $\tau=0$ for $\beta=\bar{\beta}$, as shown before. It
follows that there exists a $\beta_{0}\in(0,\bar{\beta})$ for which
we would have: $\rho_{\beta_{0}}=\sigma_{\beta_{0}}$. Let us then
fix the parameter $\beta$ by choosing that value $\beta_{0}$.

We are now in a position to prove Theorem $5.1$.
\begin{proof}

Let us use $v_{\beta_{0}}(\tau)$ as a comparison function for
$w(\tau)$. The
 differential equations to be compared are:
\begin{eqnarray*}
w''+\coth\tau w'+f'(u)w=0\\
w(0)=1, \qquad w'(0)=0
\end{eqnarray*}
and
\begin{eqnarray*}
 v''+\coth\tau v'+\left[f'(u)-\frac{\Phi(\tau)}{v}\right]v=0\\
 v(0)>0, \qquad  v'(0)=0
\end{eqnarray*}
Since in $(0,\rho)$, $\Phi<0$ and $v>0$ the coefficient of $v$ is
larger than that of $w$. Thus $v$ is a strict Sturm majorant of $w$
and its zero $\rho$ occurs before the first zero of $w$, say $c$.
But at $c$, $\Phi>0$ and $v<0$, thus the coefficient of $v$ is still
larger than that of $w$. Moreover, since $w(c)=0$,
$\displaystyle\frac{w'(c)}{w(c)}= +\infty$ and
$\displaystyle\frac{w'(c)}{w(c)}>\frac{v'(c)}{v(c)}$. Thus $v$ again
is a strict Sturm majorant of $w$. But $v$ does not have a zero in
$[c,\infty)$. Then $w$ cannot have a zero in this interval either.
So if $u\in N$ then $w(b(\alpha))<0$ and $\alpha$ is strictly
admissible. Let us consider the case $u\in G$ now. Evidently, $c$
belongs to the disconjugacy interval of (\ref{3.22}). Since $v$ is a
strict Sturm majorant of $w$ in $(0,\infty)$, the disconjugacy
interval of (\ref{3.18}) is a superset of the disconjugacy interval
of (\ref{3.22}). Thus $w$ has a zero in the disconjugacy interval of
the differential equation it satisfies. Hence it must be unbounded.

Thus for $u\in G\cup N$ the corresponding initial value is strictly
admissible.
\end{proof}
As shown before, the strict admissibility of an initial value in $G$
guarantees the uniqueness of the corresponding solution.

\begin{center}
    {\large\bf References}
\end{center}

\begin{tabular}{rp{5.5in}}

  \noindent [Bar] &   Bargmann, V.: Irreducible Unitary Representations of
the Lorentz Group, Annals of Mathematics, vol. 48, no. 3, 568-640,
  1947\\

 \noindent [Bec1] &   Beckner, W.: Sharp Inequalities and Geometric
Manifolds, The Journal of Fourier Analysis and Applications, vol. 3,
Special Issue, 825-836,
 1997\\

 \noindent [Bec2] &   Beckner, W.: Geometric Asymptotics and the
Logarithmic Sobolev Inequality, Forum Mathematicum 11, 105-137,
 1999\\

  \noindent [Ber] &   Berestycki, H., Lions, P.L.  and Peletier, L.A.: An
ODE approach to the existence of positive solutions for semilinear
problems in $R^{n}$, Indiana University Mathematics Journal 30,
141-167,
  1981\\

  \noindent [Bod] &   Bodmann, B.G.: A Lower Bound for the Wehrl Entropy
of Quantum Spin with Sharp High-Spin Asymptotics, Communications in
Mathematical Physics 250, 287-300,
  2004\\

  \noindent [Car] &   Carlen,  E.A.: Some Integral Identities and
Inequalities for Entire Functions and Their Application to the
Coherent State Transform, Journal of Functional Analysis, vol. 97,
no. 1,231-249,
  1991\\

  \noindent [Gro] &   Gross,  L.: Logarithmic Sobolev Inequalities,
American Journal of Mathematics 97, 1061-1083,
  1975\\

  \noindent [Heb] &   Hebey, E.: Nonlinear Analysis on Manifolds: Sobolev
Spaces and Inequalities, CIMS Lecture Notes, 1999, Courant Institute
of Mathematical
  Sciences\\

  \noindent [Inc] &   Ince, E.L.: Ordinary Differential Equations, Dover
Publications,
  1944\\

  \noindent [Kwo] &   Kwong, M.K.: Uniqueness of Positive Radial Solutions
of  $\triangle u-u+u^{p}=0$ in $\mathbb{R}^{n}$, Archive for
Rational Mechanics and Analysis 105, 243-266,
  1989\\

  \noindent [Lie] &   Lieb, E.H.: Proof of an Entropy Conjecture of Wehrl,
Communications in Mathematical Physics 62, 35-41,
  1978\\

  \noindent [Per] &   Perelomov, A.: Generalized Coherent States and Their
Applications, Texts and Monographs in Physics,
  Springer-Verlag\\

   \noindent [Rot] &   Rothaus, O.: Diffusion on Compact Riemannian
Manifolds and Logarithmic Sobolev Inequalities, Journal of
Functional Analysis 42, 358-367,
   1981\\

  \noindent [Sch] &   Schupp,  P.: On Lieb's Conjecture for the Wehrl
Entropy of Bloch Coherent States, Communications in Mathematical
Physics 207, no. 2, 481-493,
  1999\\

  \noindent [Weh] &   Wehrl, A.: On the Relation between Classical and
Quantum-mechanical Entropy, Reports on Mathematical Physics, vol.
16, no.3, 353-358,
  1979\\
\end{tabular}

\end{document}